\documentclass[12pt,a4paper,reqno]{article}
\usepackage{graphics}
\usepackage{epsfig}

\usepackage{amsmath,amsthm,amssymb}
\newtheorem{theorem}{Theorem}[section]

\newtheorem{example}[theorem]{Example}
\theoremstyle{definition}

\theoremstyle{remark}

\numberwithin{equation}{section}

\begin{document}
\title{ Linear Codes from Incidence Matrices  of  Unit Graphs}

\author{N. Annamalai\\
Assistant Professor\\
Indian Institute of Information Technology Kottayam\\
Pala, Kerala, India\\
{Email: algebra.annamalai@gmail.com}
\bigskip\\
C Durairajan\\
Associate  Professor\\
Department of Mathematics\\ 
School of Mathematical Sciences\\
Bharathidasan University\\
Tiruchirappalli-620024, Tamil Nadu, India\\
{Email: cdurai66@rediffmail.com}
\hfill \\
\hfill \\
\hfill \\
\hfill \\
{\bf Running head:} Linear Codes from Incidence Matrices  of  Unit Graphs}
\date{}
\maketitle

\newpage

\vspace*{0.5cm}
\begin{abstract} 
In this paper,	we examine the binary linear codes with respect to Hamming metric from incidence matrix of a unit graph $G(\mathbb{Z}_{n})$ with vertex set is $\mathbb{Z}_{n}$ and two distinct vertices $x$ and $y$  being adjacent if and only if $x+y$ is unit. The main parameters of the codes are given.
\end{abstract}
\vspace*{0.5cm}


\vspace{0.5cm}
\vspace{1.5cm}
\section{Introduction}

Let $n$ be a positive integer and let $\mathbb{Z}_n$ be the ring of integers modulo $n.$ In \cite{gri}, Ralph P. Grimaldi defined a graph $G(\mathbb{Z}_n)$ based on the elements and units of $\mathbb{Z}_n.$ The vertices of the unit graph $G(\mathbb{Z}_n)$ are the elements of $\mathbb{Z}_n$ and distinct vertices $x$ and $y$ are defined to be adjacent if and only if $x + y$ is a unit of $\mathbb{Z}_n.$ That is, $xy$ is an edge if and only if $x + y$ is a unit in $\mathbb{Z}_n.$ For a positive integer $m,$ it follows that $G(\mathbb{Z}_{2m})$ is a $\phi(2m)$-regular
graph where $\phi$ is the Euler phi function.

Codes generated by the incidence matrices and the adjacency matrices of 
various  graphs are discussed in \cite{3, cd1, cd2, anna}.

In what follows, all rings are associative with nonzero identity, denoted
by 1, which is preserved by ring homomorphisms and inherited by subrings.
Also throughout the article, by a graph $G$ we mean a finite undirected graph
without loops or multiple edges (unless otherwise specified).

\section{Preliminaries }
In this section, we study the basic properties of unit graphs and linear codes.

\indent Let $\mathbb{F}_q$ denote the finite field with cardinality $q.$ Then $\mathbb{F}_q^n$ is  a $n$-dimensional vector space over the finite field $\mathbb{F}_q.$  The Hamming weight 
$w_{H}(a$) of $a\in \mathbb{F}_q$ is \begin{align*}
	w_{H}(a)=\begin{cases}0&\,\,\text{ for }\,\,a=0\\1&\,\,\text{ for }\,\, 
		a\neq 0.
	\end{cases}
\end{align*}
Let  $a,b\in\mathbb{F}_q$, then the Hamming distance of $a$ and $b$ is defined to be $d_H(a,b)=w_{H}(a-b).$
 For a word $x=(x_1, \dots, x_n)\in \mathbb{F}_q^n,$ the Hamming weight $w_{H}(x)$ is defined to be $$w_{H}(x)=\sum\limits_{i=1}^nw_{H}(x_{i}).$$
Let $x=(x_1,\dots, x_n), y=(y_1, \dots, y_n) \in \mathbb{F}_q^n,$ then the  Hamming distance between $x$ and $y$ is defined by
$$d_H(x,y)=\sum\limits_{i=1}^{n}d_H(x_i,y_i)=\sum\limits_{i=1}
^{n}w_H(x_i-y_i).$$

A non-empty subset $C$ of $\mathbb{F}_{q}^{n}$ is said to be a {\it $q$-ary code of length $n.$} An element of the code
$C$ is called a { \it codeword.} A $q$-ary linear code $C$ of length $n$  is a 
subspace of the 
vector space $\mathbb{F}_{q}^{n}$ over 
$\mathbb{F}_{q}.$\\ 
\indent The minimum Hamming distance of a code $C$ is defined by
$$
d_{H}(C)=\min\limits_{c_{1},c_{2}\in C}\{
d_{H}(c_{1}, 
c_{2}) \mid c_1\neq c_2\}.$$
\indent The minimum weight of a code $C$ is the smallest among all weights of the non-zero codewords of $C.$ For $q$-ary linear code, we have 
$d_{H}(C)=w_{H}(C).$

All the codes here are linear codes and the notation $[n, k, d]_{q}$ 
will be used for a $q$-ary code of length $n,$ dimension $k$ and minimum 
distance $d.$ A linear $[n, k, d]_{q}$  code  is said to be a {\it Maximum Distance Separable(MDS) code} if $d=n -k + 1.$ A generator matrix $B$ for a 
linear code $C$ is a matrix whose rows form a basis for the subspace $C$ and $C_{q}(B)$ is a code generated by matrix $B$ over a finite field $\mathbb{F}_{q}.$

 Let $G=(V, E)$  be a graph with vertex set $V,$ edge set $E$ and for any $x, y \in V,\,[x,y]$ is denoted by an edge between $x$ and $y.$ An incidence matrix of 
$G$ is a $|V|\times|E|$ matrix $H$ with rows labelled by the vertices and columns by the edges and entries $h_{ij}=1$ if $i$th vertex is adjacent to the $j$th edge and $h_{ij}=0$ otherwise.

Let $R$ be a ring and let $U(R)$ be the set of unit elements of $R.$ The unit graph of $R,$ denoted $G(R),$ is the graph obtained by setting all the elements of $R$ to be the vertices and defining distinct vertices $x$ and $y$ to be
adjacent if and only if $x + y \in U(R).$

Let $G$ be a graph. We denote $V(G)$ as the vetex set of $G$ and  $E(G)$ as the edge set of $G.$ Let $x\in V(G),$ then the degree of $x,$ denoted $deg(x),$ is the number of edges of $G$ incident with $x.$ A graph $G$ is said to be a { \it $r$-regular } if the degree of each vertex of $G$ is equal to $r.$

Let $W, X\subseteq V$ with $W\cap X=\emptyset$ and let  $E(W, X)$ be the set of edges that have one end in $W$ and the other end in $X.$ Write $|E(W, X)|=q(W, X).$ 

An edge-cut of a connected graph $G$ is the set $S\subseteq E$ such that $G- S=(V, E-S)$ is disconnected.

The edge-connectivity $\lambda(G)$ is the minimum cardinality of an edge-cut. That is, $$\lambda(G)=\min\limits_{\emptyset \neq W\subsetneq V} q(W, V-W).$$
For any connected graph $G,$ we have $\lambda(G)\leq \delta(G)$
where $\delta(G)$ is minimum degree of the graph $G.$ 
\begin{theorem}\cite{ash}\label{thm1}
Let $R$ be a finite ring. Then the following statements hold for the
unit graph of $R.$
\begin{enumerate}
\item[(a)] If $2\notin U(R),$ then the unit graph $G(R)$ is a $|U(R)|$-regular graph.
\item[(b)] If $2\in U(R),$ then for every $x\in U(R)$ we have $deg(x)= |U(R)|- 1$ and for every $x\in R\setminus U(R)$ we have $deg(x)=|U(R)|.$
\end{enumerate}
\end{theorem}

\begin{theorem}\cite{3} \label{thm2}
 Let $G= (V, E)$ be a connected graph and let  $H$ be a $|V|\times|E|$ incidence matrix for $G.$ Then
		\begin{itemize}
		\item[1.] the code $C_2(H)$ is a $[|E|, |V|- 1, \lambda(G)]_2$ code
			\item[2.]  the code
		$C_p(H)$ is $[|E|, |V|- 1, \lambda(G)]_p$ code for an odd prime $p$  and  $G$ is bipartite.
	\end{itemize}
\end{theorem}
In this paper, we obtain a linear codes from the incidence matrix of the unit graph 
 $G(\mathbb{Z}_n)$ over $\mathbb{F}_2$ and we determined the main parameters of the code. In section 3, we discussed codes from incidence matrix of a unit graph $G(\mathbb{Z}_{p_1})$ over the field $\mathbb{F}_2.$ In section 4, we discussed codes from incidence matrix of a unit graph $G(\mathbb{Z}_{2p_1})$ over the finite field $\mathbb{F}_2.$ All the unit graphs considered in this article is a simple and undirected graph.
\section{Linear Codes from the Incidence Matrix of a unit Graph $G(\mathbb{Z}_{p_1})$}
In this section, we  study a binary linear code obtained from the incidence matrix of the unit graph $G(\mathbb{Z}_{p_1})$ where $p_1$ is an odd prime and find its parameters.

 The units of $\mathbb{Z}_{p_1}$  is $U(\mathbb{Z}_{p_1})=\{1, 2, \cdots, p_1-1\}.$
Let $G(\mathbb{Z}_{p_1})$ be the unit graph with vertex set $V=\mathbb{Z}_{p_1}$ and two distinct vertices $x$ and $y$ are adjacent if and only if $x+y\in U(\mathbb{Z}_{p_1}).$ 
\begin{theorem}\label{thm3}
Let $G(\mathbb{Z}_{p_1})=(V, E)$ be a unit graph. Then the graph is connected with	$|V|=p_1$ and $|E|=\dfrac{(p_1-1)^2}{2}.$
\end{theorem}
\begin{proof}
Let $G(\mathbb{Z}_{p_1})$ be a unit graph. Then by definition, $V=\mathbb{Z}_{p_1}.$ Since all non-zero elements are units, $0$ and $y$ are adjacent for all nonzero $y \in \mathbb{Z}_{p_1}.$ Therefore the graph is a connected graph.
\begin{align*}
 \Big|\Big\{[0, y] \mid y\in \mathbb{Z}_{p_1}\setminus \{0\} \Big\}\Big|&=p_1-1\\
\Big|\big\{[j, y] \mid y\in \mathbb{Z}_{p_1}\setminus \{0, 1, \cdots,j, p_1-j\} \big\}\Big|&=p_1-(j+2) \text{ for } 1\leq j\leq \frac{p_1-1}{2}
\end{align*}
For $\dfrac{p_1+1}{2}\leq j\leq p_1-2,$ we have 
$\Big|\big\{[j, y] \mid j<y \big\}\Big|=p_1-(j+1).$
Therefore,  $$|E|=p_1-1+\sum\limits_{j=1}^{\dfrac{p_1-1}{2}}p_1-(j+2)+\sum\limits_{j=1}^{\dfrac{p_1-3}{2}}j=\frac{(p_1-1)^2}{2}.$$ Hence, the unit graph $G(\mathbb{Z}_{p_1})$  is a connected graph with $p_1$ vertices and $\dfrac{(p_1-1)^2}{2}$ edges.
\end{proof}

\begin{theorem}\label{thm4}
	Let $G(\mathbb{Z}_{p_1})$ be the unit graph. Then the edge-connectivity $\lambda(G(\mathbb{Z}_{p_1}))$ of the unit graph $G(\mathbb{Z}_{p_1})$ is $p_1-2.$
\end{theorem}
\begin{proof}
If $W=\{0\}\subset V,$ then $q(W, V-W)=p_1-1.$ If $W=\{a\}\subset V, a\neq 0,$ then $q(W, V-W)=p_1-2.$ 
If $W\subset V$ with $ |W|>1,$ then $q(W, V-W)>p_1-2.$ Therefore, $\lambda(G(\mathbb{Z}_{p_1}))=p_1-2.$
\end{proof}
As a consequence of above theorems, we have 
\begin{theorem}\label{thm5}
Let $p_1$ be an odd prime. Then the code generated by the incidence matrix $H$ of the unit graph $G(\mathbb{Z}_{p_1})$ is a  $C_2(H)=\Big[\dfrac{(p_1-1)^2}{2}, p_1-1, p_1-2\Big]_2$ code over the finite field $\mathbb{F}_2.$
\end{theorem}

\begin{proof}
	Let  $G(\mathbb{Z}_{p_1})$ be a unit graph and let $H$ be the incidence matrix of $G(\mathbb{Z}_{p_1}).$ Since  $G(\mathbb{Z}_{p_1})$ is a connected graph, by Theorem \ref{thm2}, the code $C_2(H)$ is a $[|E|, |V|-1, \lambda(G(\mathbb{Z}_{p_1}))]_2$ code. By Theorem \ref{thm3}, we have  $|E|=\dfrac{(p_1-1)^2}{2}$ and $|V|=p_1.$ By Theorem \ref{thm4}, the edge-connective of the unit graph $G(\mathbb{Z}_{p_1})$ is $p_1-2$ and hence $C_2(H)$ is a $\Big[\dfrac{(p_1-1)^2}{2}, p_1-1, p_1-2\Big]_2$ code.
\end{proof}
\begin{example}
	The unit graph $G(\mathbb{Z}_{5})=(V, E)$ with $|V|=5$ and $|E|=\frac{(5-1)^2}{2}=8$ is 
	\begin{center}
		\includegraphics{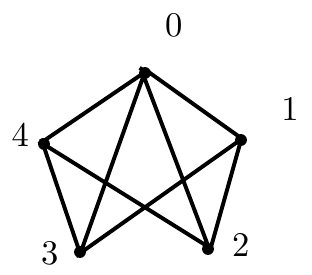}
	\end{center}
		Then the incidence matrix of the unit graph $G(\mathbb{Z}_{5})$ is
	$$H=\begin{pmatrix}
		1&1&1&1&\vline&0&0&0&0\\
		\hline
		1&0&0&0&\vline&1&1&0&0\\
		0&1&0&0&\vline&1&0&1&0\\
		0&0&1&0&\vline&0&1&0&1\\
		0&0&0&1&\vline&0&0&1&1
	\end{pmatrix}_{5\times 8}$$
Any four rows of $H$ are linearly independent and the graph is connected, then the dimension of the code $C_2(H)$ over the finite field $\mathbb{F}_2$ generated by $H$ is $4.$ The minimum distance of the code $C_2(H)$ is $3.$ Hence $C_2(H)$ is an $[n, k, d]_2=[8, 4, 3]_2$ code.
\end{example}

\section{Linear Codes from the Incidence Matrix of a unit Graph $G(\mathbb{Z}_{2p_1})$}
In this section, we  study the linear codes obtained from the incidence matrix of the unit graph $G(\mathbb{Z}_{2p_1})$ where $p_1$ is an odd prime over $\mathbb{F}_2$ and we find the parameters of the code.

The units of $\mathbb{Z}_{2p_1}$  is $U(\mathbb{Z}_{2p_1})=\{k\in \mathbb{Z}_{2p_1} \mid (k, 2p_1)=1\}$ and $|U(\mathbb{Z}_{2p_1})|=p_1-1.$

\begin{theorem}\label{thm4}
	Let $G(\mathbb{Z}_{2p_1})=(V, E)$ be a unit graph. Then 	$|V|=2p_1$ and $|E|=p_1(p_1-1).$
\end{theorem}
\begin{proof}
	Let $G(\mathbb{Z}_{p_1})$ be a unit graph. Then by definition, $V=\mathbb{Z}_{2p_1}.$ 
	Since $(2, 2p_1)=2\neq 1,$ by Theorem \ref{thm1}, $G(\mathbb{Z}_{2p_1})$ is a $|U(\mathbb{Z}_{2p_1})|=(p_1-1)$-regular graph. That is, degree of every vertex is $p_1-1.$
	Since the number of edges of the $k$-regular graph with $n$ vertices is $\dfrac{nk}{2},$ implies $|E|=p_1(p_1-1).$
\end{proof}

\begin{theorem}\label{thm6}
	Let $p_1$ be an odd  prime. Then the code generated by the incidence matrix $H$ of the unit graph $G(\mathbb{Z}_{2p_1})$ is a  $C_2(H)=\Big[p_1(p_1-1), 2p_1-1, p_1-1\Big]_2$ code over the finite field $\mathbb{F}_2.$
\end{theorem}
\begin{proof}
	Let  $G(\mathbb{Z}_{2p_1})$ be a unit graph and let $H$ be the incidence matrix of $G(\mathbb{Z}_{2p_1}).$ Since  $G(\mathbb{Z}_{2p_1})$ is a connected graph, by Theorem \ref{thm2}, the code $C_2(H)$ is a $[|E|, |V|-1, \lambda(G(\mathbb{Z}_{2p_1}))]_2$ code. Since  $G(\mathbb{Z}_{2p_1})$ is $(p_1-1)$-regular graph, the edge-connectivity $\lambda(G(\mathbb{Z}_{2p_1}))$  of the unit graph $G(\mathbb{Z}_{2p_1})$ is $p_1-1$ and hence the minimum distance of the code generated by $H$ is $p_1 -1.$  Since  $|E|=p_1(p_1-1)$ and $|V|=2p_1,$ the main parameters of the code $C_2(H)$ is  $[p_1(p_1-1), 2p_1-1, p_1-1]_2.$ 
\end{proof}
\begin{example}
	The unit graph $G(\mathbb{Z}_{6})=(V, E)$ with $|V|=6$ and $|E|=6$ is 
	\begin{center}
		\includegraphics{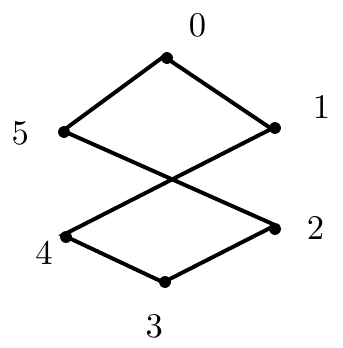}
	\end{center}
	Then the incidence matrix of the unit graph $G(\mathbb{Z}_{6})$ is
	$$H=\begin{pmatrix}
		1&0&0&0&0&0&1\\
		1&1&0&0&0&0&0\\
		0&0&1&1&0&0&0\\
		0&0&0&1&1&0&0\\
		0&1&0&0&0&1&0\\
		0&0&1&0&0&0&1
	\end{pmatrix}_{6\times 6}$$
	Any five rows of $H$ are linearly independent and the graph is connected, then the dimension of the code $C_2(H)$ over the finite field $\mathbb{F}_2$ generated by $H$ is $5.$ The minimum distance of the code $C_2(H)$ is $2.$ Hence $C_2(H)$ is an $[n, k, d]_2=[6, 5, 2]_2$ code. Since $d=n-k+1,$ the code $C_2(H)$ is a MDS code.
\end{example}

\section*{Conclusion}
In this paper, we studied the codes generated by the incidence matrix of the unit graph of the different commutative rings with unity. Also we found the main parameters of the code over finite field $\mathbb{F}_2.$ We have consider only simple and undirected graphs in this article. Finding the covering radius of these codes is the further direction to work.


\begin{thebibliography}{0}
	\bibitem{anna} N. Annamalai and C. Durairajan, Codes from the Incidence Matrices of a zero-divisor Graphs, arXiv:2011.01602, 2020.
		
	\bibitem{ash} N. Ashrafi, H. R. Maimani, M. R. Pournaki, and S. Yassemi, Unit Graphs Associated with Rings, Communications in Algebra, 38, 2851–2871, 2010.
	
	\bibitem{beck} I. Beck, Coloring of commutative rings, Journal of Algebra., 116, 208-226 (1988). 
	
	\bibitem{3}
	P. Dankelmann, J. D. Key, B. G. Rodrigues, Codes from incidence matrices of 
	graphs. Des. Codes Cryptogr. 68,1-21(2011).
	
	\bibitem{gri} P. R. Grimaldi, Graphs from rings, Proceedings of the 20th Southeastern Conference on Combinatorics, Graph Theory, and Computing (Boca Raton, FL,	1989). Congr. Numer. Vol. 71, pp. 95–103, 1990.

\bibitem{24}
	San Ling, Chaoping Xing, Coding Theory A First Course. Cambridge 	Univ. Press, New York (2004).
	
	\bibitem{cd1} R. Saranya and C. Durairajan, Codes from Incidence Matrix of some Regular Graphs, Discrete Mathematics, Algorithms and Applications, DOI: 10.1142/S179383092150035X. 
	
	\bibitem{cd2} R. Saranya and C. Durairajan, Codes from incidence matrices of $(n, 1)$-arrangement graphs and $(n, 2)$-arrangement graphs,  Journal of Discrete Mathematical Sciences and Cryptography, DOI: 10.1080/09720529.2019.1681674.

\end{thebibliography}
\end{document}